\documentclass[11pt]{article}

\usepackage[a4paper,margin=1in]{geometry}

\usepackage[utf8]{inputenc}
\usepackage[T1]{fontenc}
\usepackage{lmodern}
\usepackage{amsmath,amssymb,amsfonts,amsthm}
\usepackage{bm}
\usepackage{mathrsfs}
\usepackage{physics}
\usepackage{braket}
\usepackage{microtype}
\usepackage{indentfirst}
\usepackage{graphicx}
\usepackage{subcaption}
\usepackage{booktabs}
\usepackage{array}
\usepackage[numbers,sort&compress]{natbib}
\usepackage{microtype}

\usepackage[
colorlinks=true,
citecolor=blue,
linkcolor=blue,
urlcolor=blue
]{hyperref}

\newtheorem{theorem}{Theorem}
\newtheorem{lemma}{Lemma}

\newtheorem{conjecture}{Conjecture}

\begin{document}

\title{Volume Law and Universality of Entanglement Entropy in Random Graph Fermi Systems}

\author{Saikat Sur\thanks{Email: \texttt{saikats@imsc.res.in}} \\ 
Optics \& Quantum Information Group \\ 
The Institute of Mathematical Sciences, HBNI \\ 
CIT Campus, Taramani, Chennai 600113, India}

\maketitle
\begin{abstract}
We study the ground-state entanglement 
entropy of free fermions on random 
graphs. We first establish a general 
criterion for the volume 
law on random graphs, based on the spectral characteristics of the Hamiltonian. We then apply this 
criterion to the Erd\H{o}s--R\'enyi 
random graph, where each of the 
possible edges is present independently 
with some probability. Using random 
matrix theory and asymptotic freeness, 
we show that the ground-state entanglement 
entropy obeys an exact volume law in Erd\H{o}s--R\'enyi 
graphs in
the thermodynamic limit, with a 
universal coefficient that is 
independent of the edge probability  of the 
graph. This coefficient is confirmed 
numerically to take the value 
approximately $0.3863$ nats, strictly 
below the Page value.The volume law 
therefore reflects the absence of 
geometric locality in the random graph.
\end{abstract}
\vspace{0.5cm}

\maketitle

 \section{Introduction}
 \label{sec:intro}

The von Neumann entropy of a subsystem, 
also known as the entanglement entropy, 
is the canonical measure of bipartite 
entanglement in quantum many-body 
physics~\cite{Amico2008, 
Eisert2010, nielsen2000book}. Given a pure many-body 
state and a bipartition of the system 
into two complementary subsystems, it quantifies the degree to 
which neither half can be assigned a 
definite quantum state independently 
of the other. A vanishing entanglement 
entropy indicates that the global state 
factorizes across the bipartition. A nonzero 
and growing entanglement entropy, by 
contrast, connotes that quantum 
correlations are distributed across 
the bipartition, and its scaling with 
subsystem size  encodes 
fundamental information about the 
structure of correlations in 
the ground state~\cite{calabrese_jsm_2004,  cramer_prl_2007, laflorencie_pr_2016}.

While the entanglement entropy of 
ground states has been extensively 
studied in translationally invariant 
systems, disordered 
lattices~\cite{Gioev2006, elgart_jsp_2017, 
nandkishore_2015, abanin_2019}, and 
continuum Fermi gases~\cite{Leschke2014}, comparatively 
little is known about free fermions on 
large random graphs, where translation 
invariance and geometric locality are 
both absent. For quantum systems defined 
on graph families with large automorphism 
groups, nontrivial upper bounds on 
entanglement entropy arise from symmetry 
and representation-theoretic 
constraints~\cite{prakash2024, 
sur2026automorphism}, but such 
constraints are absent for asymmetric 
graphs. For free fermions, the 
entanglement entropy is particularly 
tractable: it is completely determined 
by the eigenvalues of the reduced 
correlation matrix~\cite{Peschel2003, peschel_jpa_2009}, 
which inherits its spectral and 
eigenvector statistics directly from 
the random Hamiltonian. The 
entanglement entropy is therefore a 
direct probe of the random matrix 
structure of the Hamiltonian. This 
raises the central question of this 
work: does the ground-state 
entanglement entropy of free fermions 
on a large asymmetric random graph 
obey a volume law, and if so, how does 
the universal coefficient depend on 
the connectivity of the graph?

{\color{black} Another open problem in the theory 
of quantum entanglement is to identify 
general spectral conditions on the 
Hamiltonian that are both necessary 
and sufficient for the ground-state 
entanglement entropy to obey a volume 
law. For translationally invariant 
free-fermionic systems, the connection 
between the Fermi surface geometry and 
the entanglement scaling is well 
established~\cite{Gioev2006, Hastings2007,eisert_rmp_2010}. 
For disordered and interacting systems, 
partial results are available~\cite{nandkishore_2015, 
abanin_2019}, while thermalizing phases 
are expected to exhibit volume-law 
entanglement consistent with the 
eigenstate thermalization 
hypothesis~\cite{deutsch_eth_2018, srednicki_pre_1994, rigol2008thermalization}. However, 
these results are specific to 
particular geometries, that do not extend 
to general random graph Hamiltonians.  
To the best of our knowledge, no 
general spectral criterion for the 
volume law of free fermions on general 
graphs has been established in the 
literature, although volume laws have 
been proved for specific cases 
including free fermions with a full 
random matrix Hamiltonian~\cite{pastur_2024} 
and structured graphs~\cite{gori_prb_2015}
What is  needed is a criterion in terms of the spectral and 
eigenvector statistics of the 
Hamiltonian, without reference to 
spatial geometry. In the present work, we formulated a criterion 
 in terms of the 
Green function of the system, that applies to 
any free-fermionic Hamiltonian on any 
random graph ensemble. It identifies 
the boundary response between the two halves of the system as the 
fundamental determinant of the entanglement scaling.}

After developing a general formalism, we investigate the ground-state entanglement entropy of free fermions on Erd\H{o}s--R\'enyi random graphs. Erd\H{o}s--R\'enyi graphs  possess a trivial symmetry in the thermodynamic limit~\cite{erdos_1959, erdos_1963, luczak_1988}, and consequently provide a natural setting  to investigate the behaviour entanglement entropy in bulk scale. Such asymmetric interaction geometries arise naturally in complex large scale quantum networks~\cite{abrahams_ws_2010,evers_rmp_2008, perseguers2010,acin2007,perseguers2010,wehner_2018,saikat_npj, biamonte_cp_2019,hens_np_2019, saikat_chaos_2026}.  Understanding the entanglement behaviour in this setting therefore offers a basic insight into the generic entangling capacity of complex quantum connectivity beyond the constraints imposed by spatial symmetries.  In the dense regime,  the edge probability $p$  remains finite as the system size $N$ tends to infinity, so that the average degree  grows linearly with system size. For the sparse regime ($p \sim 1/N$), existence of a extensive entropy density is not guaranteed~\cite{khorunzhy_jmp_1996,Semerjian2002, Bordenave2010} and  
for the threshold regime $(p \sim (\log~N)/N$),  the graph is 
connected with finite probability but the thermodynamic limit is delicate. The dense 
regime is therefore the natural regime for studying thermodynamic entanglement scaling, as the graph remains connected with high probability while the average degree diverges in the large-$N$ limit.

This paper is 
organised as follows. 
Section~\ref{sec:model} introduces 
the model of free fermions on random 
graphs and establishes the contour 
integral representation of the 
entanglement correlation matrix via 
the Schur complement identity. 
Section~\ref{sec:general} derives the 
general Green function criterion for 
the volume law, and the classification of 
random graph ensembles into Class I 
and Class II. 
Section~\ref{sec:dense} applies the 
general framework to the 
Erd\H{o}s--R\'enyi random graph, 
establishes the volume law  
and derives the universal entropy 
density. 
Section~\ref{sec:numerics} presents 
numerical results confirming the 
volume law, the universality of 
the coefficient. 
Section~\ref{sec:conclusions} 
summarises the results and discusses 
open directions.

\subsection*{Main results and assumptions}

We consider free fermion ground state 
on a random graph of $N$ vertices with a 
balanced bipartition. The  entanglement entropy 
is computed from the restricted correlation 
matrix via the Peschel formula. Our 
main results are the following.

\textbf{General criterion 
(Theorem~\ref{thm:criterion}).} For any 
free-fermionic Hamiltonian on a random 
graph with bounded spectral radius, the ground-state 
entanglement entropy obeys an exact 
volume law if  the boundary 
self-energy 
self-averages to a deterministic value
in the thermodynamic limit. All our results are based on the assumption of thermodynamic limit, that is, the system size tending to infinity.

\textbf{Volume law for the 
Erd\H{o}s--R\'enyi graph.} We prove analytically that the 
Erd\H{o}s--R\'enyi graph  at 
fixed edge probability satisfies the criterion 
of Theorem~\ref{thm:criterion} via the 
asymptotic freeness of its block 
Hamiltonians. The entanglement entropy 
therefore obeys $S_A/m \to s_\infty$ 
in the thermodynamic limit, with the $p$ independent universal coefficient $s_\infty$.
The coefficient assumes the value 
$s_\infty \approx 0.3863$ numerically. 
Based on free probability theory, we conjecture that the empirical spectral 
measure of the correlation matrix converges to the scaled
arcsine distribution, yielding 
$s_\infty = 2\ln 2 - 1$.

\section{Model and Setup}
\label{sec:model}

We consider a system of $N$ spinless noninteracting fermionic modes, one
attached to each vertex of a graph, which permits
each mode to be either empty or  occupied. The full many-body Hilbert space has dimension $2^N$, and grows exponentially with system size.
The Hamiltonian in the fermion operators is given by the form
\begin{equation}
  {h}_{N} = \sum_{i,j} (t_{N})_{ij}c_i^\dagger c_j.
  \label{eq:second_quantized}
\end{equation}
where the entry $(t_{N})_{ij}$ quantifies the hopping amplitude between
modes $i$ and $j$. The randomness enters through 
the hopping matrix $t_N$. Different random graph ensembles 
correspond to different probability distributions on 
the edge set of the Hamiltonian. The specific choice of graph 
ensemble determines the spectral statistics of $t_N$ 
and hence the entanglement properties of the ground 
state of the system. 
 
Let $\{|\phi_k\rangle\}_{k=1}^N$ denote the orthonormal
eigenstates with energy $\varepsilon_k$ of $H_{N,p}$. The eigenstates  are written in the single-particle basis
$\{|j\rangle\}_{j=1}^N$ of vertex states as $|\phi_k\rangle = \sum_{j=1}^N
\phi_k(j)|j\rangle$.
The fermionic creation operator for the eigenmode
$|\phi_k\rangle$ is the corresponding linear combination of
site creation operators,
\begin{equation}
  d_k^\dagger = \sum_{j=1}^N \phi_k(j) c_j^\dagger,
  \label{eq:dk}
\end{equation}
which creates a fermion in the energy eigenstate $|\phi_k\rangle$.
The fermion vacuum $|\Omega\rangle = \otimes^N_{k=1} |0\rangle_k$ is the state with all $N$ modes
unoccupied, satisfying $c_j|\Omega\rangle = 0$ for every
$j=1,\ldots,N$.
At half-filling with Fermi energy $\mu_F = 0$, the many-body
ground state is the Slater determinant
\begin{equation}
  |\Psi_0\rangle
  = \prod_{k:\,\varepsilon_k < 0} d_k^\dagger\,|\Omega\rangle,
\label{eq:groundstate}
\end{equation}
where $\varepsilon_k$ is the eigenvalue of $H_{N}$
associated with the eigenvector $|\phi_k\rangle$.
Since the spectrum of $H_{N,p}$ is symmetric about zero in
the large-$N$ limit, exactly $m = N/2$ modes satisfy
negative eigenvalues and are occupied with high probability.
 
For a balanced bipartition of the vertices $V = A\sqcup B$
with $|A| = |B| = N/2 \equiv m$, the entanglement between the subsystems
$A$ and $B$ in the ground state~\eqref{eq:groundstate} is captured by the projector onto the occupied
modes $\{|\phi_k\rangle : \varepsilon_k<0\}$, that built the ground state
$|\Psi_0\rangle$ out of the vacuum $|\Omega\rangle$.
Since $|\Psi_0\rangle$ is a Slater determinant, Wick's
theorem guarantees that the reduced density matrix
$\rho_A = \mathrm{Tr}_B|\Psi_0\rangle\langle\Psi_0|$
is entirely determined by the one-body correlation
function~\cite{Peschel2003}
\begin{equation}
  (C_A)_{ij}
  = \langle\Psi_0|c^\dagger_jc_i|\Psi_0\rangle,
  \quad i,j \in A,
  \label{eq:CA_def}
\end{equation}
The matrix element $\langle\Psi_0|c_j^\dagger c_i|\Psi_0\rangle$
is the projection of $P_F = \sum_{\varepsilon_k<0}|\phi_k\rangle\langle\phi_k|$ onto sites $i,j\in A$, written in
the site basis.
Since $C_A$ is a compression of the projector $P_F$, its
eigenvalues $\{\lambda_i\}_{i=1}^m$ lie in $[0,1]$.
Physically, $C_A$ encodes how the $m$ occupied global modes
that make up $|\Psi_0\rangle$ are redistributed when restricted to subsystem $A$.  Each eigenvalue $\lambda_i$
is the probability that the corresponding mode is occupied
in the reduced state $\rho_A$, with $1-\lambda_i$ the
probability that it is empty. Because $|\Psi_0\rangle$ is Gaussian, $\rho_A$ factorises into independent two-level systems, one per
eigenmode of $C_A$, $\rho_A = \bigotimes_{i=1}^m
\mathrm{diag}(1-\lambda_i,\lambda_i)$, so the
entanglement entropy is  the sum of the binary
 Shannon entropies of these $m$ effective
modes~\cite{Peschel2003, nielsen2000book},
\begin{equation}
  S = -\sum_{i=1}^{m}
      \bigl[\lambda_i\log\lambda_i
            + (1-\lambda_i)\log(1-\lambda_i)\bigr].
  \label{eq:entropy}
\end{equation}

The Fermi projector $P_F$ can itself be written
as a contour integral of the Green function of $H_{N,p}$.
For any Hermitian matrix $H$ with eigenpairs
$\{\varepsilon_k,|\phi_k\rangle\}$, the Green function
$G(z) = (zI-H)^{-1} = \sum_k |\phi_k\rangle\langle\phi_k|/(z-\varepsilon_k)$
has a simple pole at each eigenvalue, with residue equal to
the projector onto the corresponding eigenstate.
By the Cauchy residue theorem, integrating $G(z)$ around any
closed contour $\gamma$ therefore reproduces the sum of
exactly those eigen-projectors whose eigenvalues lie inside
$\gamma$. This contour integral
of the full $N\times N$ Green function is the Fermi
projector $P_F$. Since $C_A = P_F|_A$ is the $(A,A)$ block of $P_F$, it
suffices to know only the effective Hamiltonian of the block $(A,A)$ at each point on $\gamma$, rather than
the full $N\times N$ matrix. To compute the correlation matrix $C_A$ in \eqref{eq:CA_def} without diagonalising the full matrix
$H_{N,p}$, we write $H_{N,p}$ in the block form induced by the
bipartition,
\begin{equation}
  H_{N,p} =
  \begin{pmatrix}
    H_{AA} & H_{AB} \\
    H_{AB}^\top & H_{BB}
  \end{pmatrix},
  \label{eq:block}
\end{equation}
where $H_{AA}$ and $H_{BB}$ are $m\times m$ real symmetric matrices
of intra-subsystem hoppings and $H_{AB}$ is an $m\times m$ real
matrix of inter-subsystem hoppings.
The block structure of $H_N$ induced by the bipartition 
encodes the intra-subsystem hoppings $H_{AA}$, $H_{BB}$ 
and the inter-subsystem hoppings $H_{AB}$. The statistical 
properties of these blocks depend on the specific random 
graph ensemble under consideration.
The Schur complement identity gives an exact contour integral
representation of entanglement correlation matrix~\cite{zhang2005} (see Lemma~\ref{lem:schur} in Appendix A)
\begin{equation}
C^{(N)}_A = \frac{1}{2\pi i} \oint_{\gamma}
        G^{N}_{AA}(z) dz,
  \label{eq:CA_contour}
\end{equation}
where
\begin{equation}
  G^{(N)}_{AA}(z) = 
        \bigl(zI_m - H_{AA} - \Sigma^{(N)}(z)\bigr)^{-1},  
        \label{eq:GAA}
\end{equation}
is  the Green function of the joint operators $H_{AA} + \Sigma(z)$ at spectral parameter $z \in \mathbb{C}$, and

\begin{equation}
  \Sigma^{(N)}(z) = H_{AB}\bigl(zI_m - H_{BB}\bigr)^{-1}H_{AB}^\top
  \label{eq:Sigma}
\end{equation}
is the self-energy  from integrating out subsystem $B$. The contour $\gamma$ is a positively oriented contour in 
$\mathbb{C}\setminus\mathbb{R}$ enclosing exactly the 
negative eigenvalues of $H_N$; the existence of such a 
deterministic contour is established in 
Lemma~\ref{lem:boundary_response}. The operator $H_{\mathrm{eff}}^{(N)}(z) 
  = H_{AA}^{(N)} + \Sigma^{(N)}(z)$ is the effective Hamiltonian of subsystem $A$
after the $B$ degrees of freedom have been integrated out, with the self-energy $\Sigma^{(N)}(z)$ encoding the influence of $B$ on
$A$ mediated through the inter-subsystem hopping $H_{AB}$. 
Eq.~\eqref{eq:CA_contour} is exact for any finite $N$,
and the thermodynamic limit behaviour of the self energy is the key object
of the analysis that follows.

\section{When Does a Volume Law Emerge?}
\label{sec:general}

For the contour integral~\eqref{eq:CA_contour} to produce 
a deterministic limit in the thermodynamic 
limit, it suffices that the self-energy assumes a value 
deterministically in the thermodynamic limit. If this holds, the effective Hamiltonian  becomes deterministic, and the restricted Green function  concentrates around a deterministic 
limit uniformly on the contour $\gamma$. 
The following lemma and the theorem make this argument precise.

\setcounter{lemma}{1}

\begin{lemma}
\label{lem:boundary_response}

Let $H_N$ be a Hermitian matrix with 
balanced bipartition $V = A\sqcup B$, 
$|A|=|B|=m$, and let $\eta > 0$ be the 
spectral gap parameter,  measuring the 
minimum distance from the integration 
contour to the real axis. Suppose that 
the self-energy $\Sigma^{(N)}(z)$ 
converges to a deterministic limit 
$\bar\Sigma(z)$ uniformly for 
$|\mathrm{Im}(z)|\geq\eta$.
\begin{equation}
  \lim_{N \rightarrow \infty}\sup_{|\mathrm{Im}(z)|\geq\eta}
  \|\Sigma^{(N)}(z) - \bar\Sigma(z)\|
  = 0.
\label{eq:sigma_conv}
\end{equation}
Then $G_{AA}^{(N)}(z)$ converges to 
$\bar{G}_{AA}^{(N)}(z) = (zI_m - 
H_{AA}^{(N)} - \bar\Sigma(z))^{-1}$ 
in the thermodynamic limit. Note that $\bar{G}_{AA}^{(N)}(z)$ 
still depends on the random matrix 
$H_{AA}^{(N)}$. (See 
Appendix~A for the proof.)
\end{lemma}

\begin{theorem}
\label{thm:criterion}
Let $H_N$ be a free-fermionic Hamiltonian 
defined on a random graph with $N$ 
vertices, balanced bipartition 
$V = A\sqcup B$, $|A|=|B|=m=N/2$, 
and Fermi level $\mu_F = 0$. Suppose 
that

(i) the spectral radius is 
bounded, $\|H_N\| \leq M$ 
for some deterministic $M$ independent of $N$;

(ii) the self-energy converges 
to a deterministic limit $\bar\Sigma(z)$ 
in the thermodynamic limit.

Then the entanglement entropy satisfies 
the exact volume law
\begin{equation}
\lim_{N\to\infty}  \frac{S_A}{m} = 
  s_\infty, 
  \qquad 
  \frac{1}{m}\operatorname{Tr}(C_A-C_A^2) 
  \leq s_\infty \leq \log 2,
  \label{eq:volume_law}
\end{equation}
provided $s_\infty > 0$. (See 
Appendix~A for the proof.)
\end{theorem}

The convergence of the self energy
to a deterministic limit and boundedness of the spectral radius of the Hamiltonian
are the minimal 
conditions for a volume law to hold. This 
requires the entire product $
  \Sigma^{(N)}(z) 
  = H_{AB}^{(N)}
    (zI_m - H_{BB}^{(N)})^{-1}
    (H_{AB}^{(N)})^\top$
to self-average. It does not necessarily require 
its individual factors to self-average separately, 
since the product can converge to a 
deterministic limit even when the 
individual factors do not. A necessary 
condition for $\Sigma^{(N)}(z)$ to 
be bounded is $
  \|\Sigma^{(N)}(z)\| 
  \leq \|H_{AB}^{(N)}\|^2/\eta
  = \mathcal{O}(1)$.
This condition requires $\|H_{AB}^{(N)}\| = 
mathcal{O}(1)$. Therefore random graph ensembles can be classifies  into two classes. \textbf{Class I:} 
The self-energy self-averages to a 
deterministic limit and the spectrum of the Hamiltonian is bounded. The correlation 
matrix $C_A^{(N)}$ has a well defined 
thermodynamic limit and a volume law 
can emerge. The Erd\H{o}s--R\'enyi 
graph in the dense regime 
belongs to this class and will be discussed in detail in the next section. \textbf{Class II:} Either the self-energy diverges or 
fluctuates between graph realisations or the spectrum of the Hamiltonian is unbounded. 
The correlation matrix $C_A^{(N)}$ 
does not have a well defined 
thermodynamic limit and no volume 
law exists. The entanglement entropy  in this case
remains bounded above by the Page value, but its 
scaling in the thermodynamic limit 
is an open problem beyond the scope 
of this paper.

\begin{figure}[t]
\centering
\includegraphics[width=0.70\linewidth]{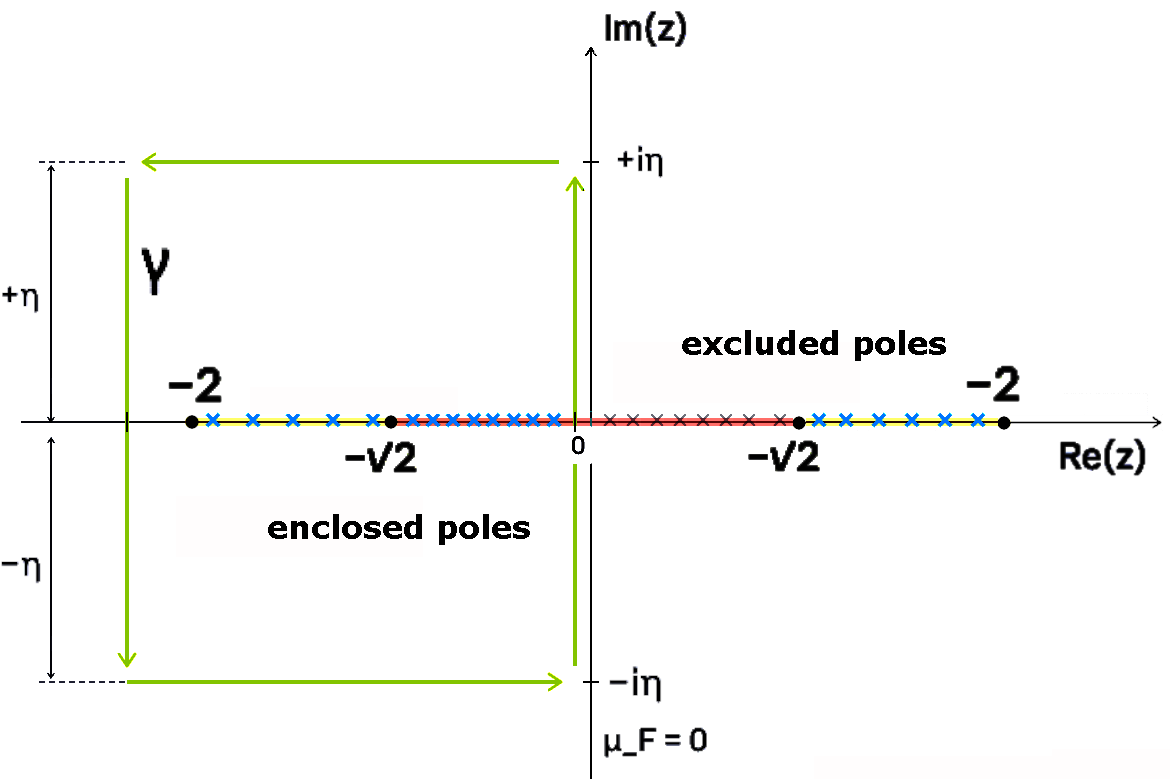}
\caption{A schematic representation of the integration contour $\gamma$ in the complex plane for 
computing the entanglement correlation matrix $C_A$ via 
Eq.~\eqref{eq:CA_contour}. The contour is the boundary of the rectangle 
$\{z = x + iy : x_L \leq x \leq x_R, -\eta \leq y \leq \eta\}$, 
traversed counter-clockwise. The boundary of 
$\gamma$ lies strictly off the real axis at imaginary offsets $\pm i\eta$. The crosses  on the real axis 
denote  the poles 
of $G_{AA}(z)$; those with $\varepsilon_k < 0$ lie strictly 
inside $\gamma$ and are picked up by the residue theorem, while those 
with $\varepsilon_k > 0$  lie outside $\gamma$ and do not 
contribute. The right boundary of the contour is chosen 
immediately to the left of the Fermi energy $\mu_F = 0$.}
\label{fig:contour}
\end{figure}

\section{Free Fermions on  Erd\H{o}s-R\'enyi 
Random Graphs}
\label{sec:dense}

Having established the general conditions 
under which a free-fermionic ground state 
on a random graph exhibits a volume law 
of entanglement entropy, we now turn to 
a concrete example that 
satisfies all these conditions exactly. 
The Erd\H{o}s--R\'enyi random graph is the 
simplest and most tractable random graph 
ensemble in which the self-energy 
$\Sigma^{(N)}(z)$ self-averages to a  deterministic limit in the thermodynamic 
limit. This follows from the statistical 
independence of the edge sets corresponding 
to $H_{AB}^{(N)}$ and $H_{BB}^{(N)}$, 
which is a direct consequence of the 
independent edge structure of the 
Erd\H{o}s--R\'enyi model. 

The Erd\H{o}s--R\'enyi random graph has $N$ vertices, in which each of the $C(N,2)$
edges is present independently with probability
$p\in(0,1)$~\cite{erdos_1959, erdos_1963}. The one-body Hamiltonian is the centred and normalized
adjacency matrix
\begin{equation}
  H_{N,p} = \frac{t_{N,p} - p(\mathbf{1}\mathbf{1}^\top - I_N)}{\sqrt{Np(1-p)}},
  \label{eq:H}
\end{equation}
where $t_{N,p}$ is the $N\times N$ adjacency matrix of graph,
$I_N$ is the identity, and $\mathbf{1} = (1,\ldots,1)^\top
\in\mathbb{R}^N$, so that $\mathbf{1}\mathbf{1}^\top - I_N$
is the all-ones matrix with zero diagonal.
After this normalisation, the off-diagonal entries have
mean zero and variance $1/N$.
For fixed $p\in(0,1)$, the matrix $H_{N,p}$ falls within
the Wigner universality class. Its empirical spectral
measure converges to the semicircle law in the
thermodynamic limit, independently of the entry
distribution, as for any Wigner-type ensemble with i.i.d.\
entries of mean zero and finite variance~\cite{Anderson2010,ErdosKnowlesYauYin2013}. The transformation from the original Hamiltonian 
$t_{N,p}$ to the normalized form $H_{N,p}$ is a rescaling that places the Fermi level 
at zero and brings the spectrum into the Wigner 
universality class. As shown in 
Appendix~\ref{app:p_independence}, this 
transformation leaves the entanglement correlation 
matrix $C_A$ unchanged in the thermodynamic limit. The bulk eigenvectors of $t_N$ and 
$H_{N,p}$ are identical, and the same set of 
eigenvectors is occupied at half-filling under both 
Hamiltonians. Consequently, the entanglement entropy 
density $s_\infty$ computed from $h_{N,p}$ is 
identical to that computed from $H_{N,p}$ in the 
thermodynamic limit.

We analyse Eq.~\eqref{eq:CA_contour} in the thermodynamic
limit where the edge probability $p$ is fixed. 
The blocks $H_{AA}$ and $H_{BB}$ are $m\times m$ Wigner-type
matrices with centred entries of variance $1/N$. Since
$m = N/2$, and by the Wigner semicircle
law~\cite{ErdosKnowlesYauYin2013} their empirical spectral
measures converge in the thermodynamic limit to the semicircle
supported on $[-\sqrt{2},\sqrt{2}]$. 
The matrix $W = H_{AB}H_{AB}^\top$ is an $m\times m$ positive
semidefinite square Wishart-type matrix whose empirical spectral measure converges
to the Marchenko--Pastur law with unit ratio supported on
$[0,2]$~\cite{marchenko_1967, bai2010} (see Appendix~\ref{app:B}). {The contour of the integral  $\gamma$ in \eqref{eq:CA_contour} is the boundary of the rectangle
  $(z = x+iy : x_L \le x \le x_R,\ -\eta \le y \le \eta )$,
with $x_L < -2$ chosen to the left of the entire spectrum of
$H_{N,p}$, $x_R = 0^-$, and $\eta>0$  is a fixed imaginary offset keeping
$\gamma$ off the real axis.
The contour is traversed counter-clockwise,  
so that every eigenvalue $\varepsilon_k<0$ of $H_{N,p}$ lies in the interior of $\gamma$ (see Fig.~\ref{fig:contour}).}

{\color{black} We approximate the Green function using the asymptotic freeness between the blocks $H_{AA}$ or $H_{BB}$ and $H_{AB}$ of the matrix in \eqref{eq:block}.
By the local semicircle law for $H_{BB}$~\cite{ErdosKnowlesYauYin2013},
the Green function of $H_{BB}$  in the thermodynamic limit tends to
\begin{eqnarray}
 \lim_{N \rightarrow \infty}&&(zI_m - H_{BB})^{-1} = g(z)I_m, \text{~~where~~~}  g(z) = z - \sqrt{z^2 - 2}
  \label{eq:gz}
\end{eqnarray}
is the Stieltjes transform of the semicircle on
$[-\sqrt{2},\sqrt{2}]$ (see Appendix~\ref{app:B}).
Since $H_{AB}$ is independent of $H_{BB}$, the exact
self-energy~\eqref{eq:Sigma} concentrates to
\begin{equation}
 \lim_{N\to\infty} \Sigma(z) = g(z)W.
  \label{eq:Sigma_largeN}
\end{equation}
The asymptotic freeness assumes that the subsystem $A$ and its environment $B$ are represented by free random matrices in the thermodynamic limit, implying the three blocks $H_{AA}$, $H_{BB}$, and $H_{AB}$ of the Hamiltonian 
$H_{N,p}$ correspond to disjoint edge sets of the Erd\H{o}s--R\'enyi 
graph.  
Therefore, the self-energy $\Sigma(z) = H_{AB}(zI_m - H_{BB})^{-1}H_{AB}^\top$ 
involves the product of $H_{AB}$ and a function of the independent matrix 
$H_{BB}$, and its thermodynamic limit is computed using asymptotic freeness 
between independent Wigner-type random 
matrices~\cite{Anderson2010, Mingo2017}. 
However, asymptotic freeness captures only the leading-order
behaviour of $\Sigma(z)$ in the thermodynamic limit, with corrections 
encoded in $\Delta(z) = \Sigma(z) - g(z)W$.

The magnitude of this correction $\Delta(z)$ is governed by the boundary-to-bulk  ratio of the graph. In an Erd\H{o}s--R\'enyi graph with balanced 
bipartition, $|A|=|B|=N/2$, each of the $(N/2)^2$ possible boundary edges 
between $A$ and $B$ is present independently with probability $p$. It yields 
an expected number of boundary edges $|E_{AB}| \sim pN^2/4$. Similarly, 
each of the $C(N/2,2)$ possible bulk edges within $B$ is present 
independently with probability $p$, giving $|E_{BB}| \sim pN^2/8$. The 
boundary-to-bulk ratio is
  $  |E_{AB}|/|E_{BB}| = 2$,
which is $\mathcal{O}(1)$ and therefore independent of both $N$ and $p$. Unlike geometrically 
local models in $d$ spatial dimensions, where the boundary-to-bulk ratio 
vanishes as $\mathcal{O}(N^{-1/d})$~\cite{Eisert2010}, the Erd\H{o}s--R\'enyi 
random graph has no geometric locality, and the boundary remains 
extensive even in the thermodynamic limit. Consequently, the correction term 
$\Delta(z)$ remains $\mathcal{O}(1)$ in the thermodynamic limit and contributes a finite 
correction to the entropy computed from the approximation of asymptotic freeness.}

After assuming this approximation of asymptotic freeness , the Green function~\eqref{eq:GAA} therefore simplifies to
\begin{equation}
\bar{G}_{AA}(z) =  \lim_{N\to\infty} G_{AA}(z) = \bigl(zI_m - H_{AA} - g(z) W\bigr)^{-1}.
  \label{eq:GAA_limit}
\end{equation}
Equation~\eqref{eq:GAA_limit} is therefore the central formula of the
dense regime. The expression is exact within the asymptotically free effective theory, in the thermodynamic limit, as it involves
only the $m\times m$ blocks $H_{AA}$ and $W$, and it requires no
diagonalization of the full Hamiltonian and the limiting Green function $\bar{G}_{AA}(z)$ is universal.
The semicircle law governing $H_{AA}$ and the Marchenko--Pastur
law governing $W$, as well as the scalar function $g(z)$, all
depend only on the entry variance $N^{-1}$, and are independent of the edge probability
$p$ and the microscopic distribution of the matrix
entries.
Consequently, in the thermodynamic limit, $\bar{G}_{AA}(z)$
depends only on these two universal spectral laws. 

It remains to 
verify that $s_\infty > 0$. If 
$s_\infty = 0$, then by the half-filling 
constraint $\int\lambda\,d\bar\rho_{C_A} 
= 1/2$, the measure $\bar\rho_{C_A}$ 
must equal $\frac{1}{2}\delta(\lambda) 
+ \frac{1}{2}\delta(\lambda-1)$, 
requiring $C_A$ to be a projector and 
the subspace $A$ to be invariant under 
$P_F$. However, by the quantum ergodicity 
of Wigner 
matrices~\cite{ErdosKnowlesYauYin2013}, 
the eigenvectors of $H_{N,p}$ are 
delocalized with $|\psi_k(i)|^2\sim 
1/N$ for all $i$, so no eigenvector 
is supported entirely on $A$ or $B$. 
Therefore $C_A$ has generic nonzero 
off-diagonal elements and $s_\infty > 0$.

The Erd\H{o}s--R\'enyi graph satisfies 
all conditions of 
Theorem~\ref{thm:criterion}. The 
self-energy self-averages to the 
deterministic limit $g(z)W$ under the 
asymptotic freeness approximation, 
as established above. It remains to 
verify that $s_\infty > 0$. If 
$s_\infty = 0$, then by the half-filling 
constraint $\int\lambda\,d\bar\rho_{C_A} 
= 1/2$, the measure $\bar\rho_{C_A}$ 
must equal $\frac{1}{2}\delta(\lambda) 
+ \frac{1}{2}\delta(\lambda-1)$, 
requiring $C_A$ to be a projector and 
the subspace $A$ to be invariant under 
$P_F$. However, by the quantum ergodicity 
of Wigner 
matrices~\cite{ErdosKnowlesYauYin2013}, 
the eigenvectors of $H_{N,p}$ are 
delocalized with $|\phi_k(i)|^2\sim 
1/N$ for all $i$, so no eigenvector 
is supported entirely on $A$ or $B$~\cite{ErdosKnowlesYauYin2013,erdos_2011_bulk,erdos_cmp_2009}. 
Therefore $C_A$ has generic nonzero 
off-diagonal elements and $s_\infty > 0$.

The Erd\H{o}s--R\'enyi graph therefore 
satisfies all conditions of 
Theorem~\ref{thm:criterion}. The 
self-energy self-averages to the 
deterministic limit $g(z)W$ under the 
asymptotic freeness approximation, 
and $s_\infty > 0$ follows from the 
quantum ergodicity of Wigner matrices 
as established in 
Section~\ref{sec:general}. The volume 
law $S_A = ms_\infty$ therefore 
holds with a universal coefficient 
$s_\infty$ independent of $p$ and the 
microscopic graph structure.

{\color{black} The maximal entropy density
would require all eigenvalues of \(C_A\) to converge to \(1/2\). Since every matrix element of \(C_A\) is a coherent sum over the occupied
eigenvectors of the random graph Hamiltonian,  the exact saturation of the
maximal entropy density would therefore require all these sums to cancel asymptotically for every pair \(i\neq j\).
The vanishing of all off-diagonal elements of $C_A$ 
simultaneously is
impossible in the Erd\H{o}s--R\'enyi ensemble, for the 
following reasons. (a) The ensemble has no geometric symmetry that could enforce systematic 
cancellation among the coherent sums 
$ (C_A)_{ij} = \sum_{\varepsilon_k<0}\phi_k(i)\phi_k(j)$. This implies that all 
$C(N/2,2)\sim N^2/8$ off-diagonal elements vanish 
simultaneously imposing $\mathcal{O}(N^2)$ independent conditions on 
the $N/2$ eigenvectors, which is an overdetermined system 
that generically has no solution. (b) By the 
quantum ergodicity of the ER random graph, the eigenvectors 
$\{\phi_k\}$ are delocalized, with components satisfying 
$|\phi_k(i)|^2 \sim 1/N$ for all $i$. 
Each term $\phi_k(i)\phi_k(j) \sim 1/N$ is therefore 
nonzero, and the coherent sum over $N/2$ occupied modes 
gives $(C_A)_{ij} = \mathcal{O}(1)$, which does not vanish in the 
thermodynamic limit. }
Therefore the inequality~\eqref{eq:Jensen} is strict, and
the \emph{entropy deficit}
\begin{equation}
  \Delta s = \log 2 - s_\infty
  = \int_0^1 \bigl[\log 2 - h(\lambda)\bigr]
    d\bar{\rho}_{C_A}(\lambda) > 0
  \label{eq:deficit}
\end{equation}
is strictly positive.
The quantity $\Delta s$ measures the spread of $\bar{\rho}_{C_A}$
around $1/2$. It vanishes only in the maximally entangled limit
$\bar{\rho}_{C_A} = \delta(\lambda - \tfrac{1}{2})$, corresponding
to the Page value $s = \log 2$, and it increases as the spectral
measure broadens away from $1/2$. Therefore, the ground state  does not achieve the Page value, and the deficit is a universal constant.

Although analytical theory establishes the existence of an extensive entanglement entropy, obtaining the universal coefficient $s_{\infty}$ analytically requires determining the spectral measure of the restricted correlation matrix $C_A$. Although the asymptotic Green function can be obtained implicitly through the Dyson equation, extracting the spectral measure of $C_A$ entails solving a nonlinear algebraic equation for its Stieltjes transform, for which no closed-form solution is currently known. So the coefficient $s_{\infty}$ should be determined numerically.

\begin{figure*}[t]
\begin{center}
\includegraphics[width=0.980\textwidth]{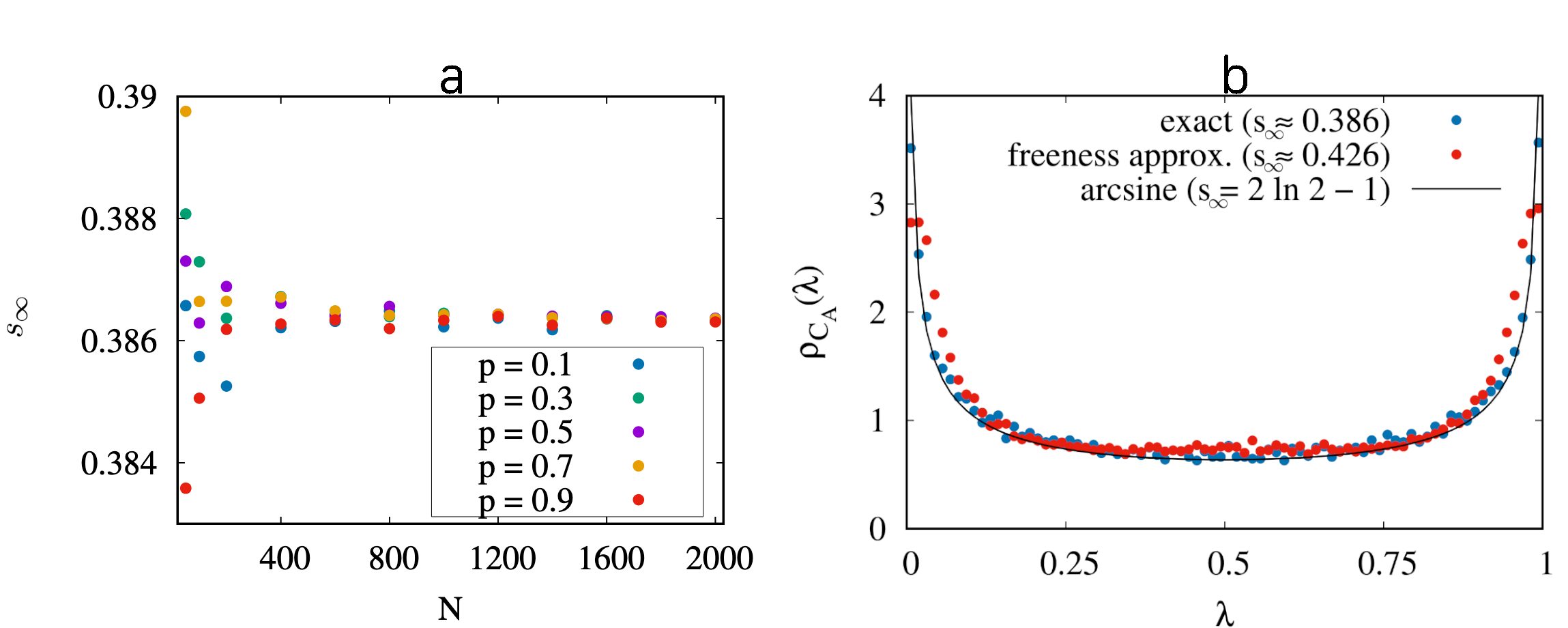}
\end{center}
\caption{(a) Entropy density $s_\infty$ as a function of system size $N$ 
for different values of the edge probability, computed via the 
exact diagonalization of \eqref{eq:second_quantized}. Each data 
point represents the average over $100$ independent 
Erd\H{o}s--R\'enyi random graph realisations. The 
entropy density is independent of $p$ for all system 
sizes considered, consistent with the 
volume law. (b) The spectral density $\bar{\rho}_{C_A}(\lambda)$ of the
entanglement correlation matrix $C_A$,  for subsystem
size $N = 500$, computed via the contour integral
formula~\eqref{eq:CA_contour} and from exact diagonalization of \eqref{eq:second_quantized}.
Both the distributions have  U-shaped profiles, with enhanced weight near
 the endpoints. The von Neumann entanglement entropy obtained from the  asymptotic freeness approximation yields $0.426$ nats, from exact diagonalization yields $0.3863$ nats, from the arcsine distribution yields $2 \ln 2 - 1 \approx 0.386$ nats. }
\label{fig:1}
\end{figure*}

\section{The Universal Coefficient }
\label{sec:numerics}

Figure~\ref{fig:1}(a) shows the entropy density 
 as a function 
of system size $N$ for different values of the edge 
probability. 
The entropy density is independent of $p$ for 
all system sizes considered, confirming the 
analytical result. The 
entropy density converges to a universal constant 
$s_\infty \approx 0.386$ nats as 
$N\to\infty$, establishing 
the exact volume law numerically. The finite-$N$ 
corrections are visible for small system sizes 
but become negligible as $N$ increases.

Figure~\ref{fig:1}(b) shows the empirical 
spectral density $\rho_{C_A}(\lambda)$ of the 
entanglement correlation matrix $C_A$ for 
$N = 500$ and $p = 0.5$, computed from the 
exact many-body ground state of  
and from the asymptotic freeness approximation 
via the contour integral~\eqref{eq:CA_contour}. 
Both distributions exhibit a U-shaped 
profile. This U-shape reflects the 
spread of eigenvalues of $C_A$ away from the 
uniform value $1/2$, which is responsible for 
the positive entropy deficit. The freeness approximation 
places slightly less weight near the endpoints than the 
exact distribution. Hence we see a difference in the entropy density 
$s_\infty^{\text{free}} \approx 0.426 > 
s_\infty \approx 0.386$ is observed.

The discrepancy between the two distributions 
in Fig.~\ref{fig:1}(b) has a physical 
origin. The asymptotic freeness approximation 
replaces the exact environment response 
$(zI_m - H_{BB})^{-1}$ by its isotropic scalar 
limit $g(z)I_m$. This approach preserves the average spectral 
density of subsystem $B$ via the Stieltjes 
transform  but discards all 
information encoded in the eigenvectors of 
$H_{BB}$ and all higher-order correlations 
between the environment modes. From the 
perspective of subsystem $A$, the environment 
is replaced by an isotropic effective medium 
that retains only the coarse-grain spectral information 
of $H_{BB}$. 
This structure breaks the non-homogeneous mixing of the
subsystem $A$ to different modes of $B$,  allowing the concentration of  the eigenvalues of $C_A$ toward $1/2$ and  therefore reducing the entropy below the 
freeness approximation value.


It is interesting to note that the numerical value $s_\infty \approx 0.386$ 
agrees, with  $
  2\ln 2 - 1 = 0.38629\ldots$, which is also the exact value of the entropy 
integral against the \emph{arcsine distribution} defined 
on $[0,1]$
\begin{equation}
  \frac{1}{\pi}\int_0^1
  \frac{h(\lambda)}{\sqrt{\lambda(1-\lambda)}}
  d\lambda = 2\ln 2 - 1.
  \label{eq:arcsine_integral}
\end{equation}
 This coincidence, together with the 
structure of $C_A$ as a compression of the Fermi 
projector, suggests a mechanism, which we 
state as a conjecture (see Fig.~\ref{fig:1}b).

The correlation matrix is exactly $C_A = P_AP_FP_A|_A$, 
a compression of the Fermi projector 
$P_F=\mathbf{1}[H_{N,p}<0]$ by the deterministic 
subsystem projector $P_A$, both of trace $1/2$. 
Asymptotic freeness of Wigner-type matrices from 
deterministic diagonal matrices, such as $P_A$, is 
classical~\cite{dykema1992, Anderson2010}, and is expected to extend from 
the Hamiltonian $H_{N,p}$ itself to its spectral 
projector $P_F$, given the delocalization and spectral 
rigidity of $H_{N,p}$ near the Fermi 
level~\cite{ErdosKnowlesYauYin2013}. 
{\color{black} If $P_A$ and $P_F$ are asymptotically free, then the matrices
$ X_A = 2P_A-I$, and $
  X_F = 2P_F-{I}$,
satisfy $X_A^2=X_F^2={I}$ and therefore each has the
symmetric Bernoulli spectral measure
\begin{equation}
  \mu=\frac12(\delta_{-1}+\delta_{1}).
\end{equation}
Since freeness is preserved under the transformations,
$X_A$ and $X_F$ are themselves free. Voiculescu's 
result that the free additive convolution of two symmetric
Bernoulli measures is the arcsine law~\cite{Voiculescu1991},
\begin{equation}
  \mu_B\boxplus\mu_B
  =
  \frac{dx}{\pi\sqrt{4-x^2}},
  \qquad x\in[-2,2],
\end{equation}
is equivalent to the free Jacobi law for two free projections of
trace $1/2$~\cite{nica2006}. Consequently,
the spectral measure of the compressed correlation matrix
$C_A=P_AP_FP_A|_A$ is
\begin{equation}
  \bar\rho_{C_A}(\lambda)
  =
  \frac{1}{\pi\sqrt{\lambda(1-\lambda)}},
  \qquad
  \lambda\in[0,1],
  \label{eq:arcsine_conjecture}
\end{equation}
with no atoms at the endpoints.}

\begin{conjecture}
\label{conj:arcsine}
In the thermodynamic limit, the empirical spectral 
measure of the entanglement correlation matrix $C_A$ 
for the Erd\H{o}s--R\'enyi graph at half-filling 
converges to the arcsine 
distribution~\eqref{eq:arcsine_conjecture}, and the 
entropy density equals $s_\infty = 2\ln 2 - 1$, 
independent of $p$.
\end{conjecture}

This mechanism also connects our ground-state result 
to the average eigenstate entanglement entropy of 
random quadratic Hamiltonians at subsystem fraction 
$1/2$, where the same 
constant $2\ln 2-1$ has been obtained by 
entirely different methods~\cite{lydzba_2020}. 
The agreement suggests that the ground state of 
$H_{N,p}$ is, in this precise sense, a typical 
eigenstate of a quantum-chaotic quadratic Hamiltonian. Conjecture~\ref{conj:arcsine} makes 
predictions beyond the entropy itself. 
However,  technical gaps remain before 
Conjecture~\ref{conj:arcsine} can be upgraded to a 
theorem.  Asymptotic freeness must be extended 
from polynomials in $H_{N,p}$ to its spectral 
projector $P_F$, which is discontinuous.Although local semicircle laws and eigenvalue rigidity
provide the natural analytical framework, we are not
aware of existing results that directly establish these
properties. A rigorous proof  of the conjecture may require techniques beyond those currently
available, and we therefore leave them for future work.

{\color{black}\section{Discussion and Conclusions}
\label{sec:conclusions}

The contour integral representation of the  correlation matrix and
the Schur complement identity, provides a 
general framework for understanding  the scaling behavior of the 
ground-state entanglement entropy. The three exact 
equations~\eqref{eq:CA_contour}--\eqref{eq:Sigma} 
hold for any free fermion Hamiltonian with a balanced 
bipartition, independently of the graph 
structure. The 
entanglement entropy satisfies an exact volume 
law 
if two conditions hold in the thermodynamic 
limit: (a) the empirical spectral measure 
$\rho_{C_A}$ converges  to 
a deterministic limit $\bar{\rho}_{C_A}$, in the thermodynamic limit as a consequence of self energy self averaging to a deterministic value and the spectrum of the Hamiltonian being bounded; and (b)
the limiting entropy density $s_\infty$ 
strictly positive. The first condition holds 
automatically whenever the empirical spectral 
measures of the diagonal blocks $H_{AA}$ and 
$H_{BB}$ converge to deterministic limits and 
the off-diagonal block $H_{AB}$ has bounded 
operator norm. The second condition requires the subsystem 
$A$ not to be invariant under the 
Fermi projector, which holds 
whenever the eigenstates of the 
Hamiltonian are sufficiently 
delocalized across the bipartition. Random graph ensembles fall into two 
classes: Class I, where both conditions 
hold and a volume law emerges, and 
Class II, where one or both conditions 
fail. Geometrically local graphs, 
where the boundary-to-bulk ratio 
vanishes, and scale-free graphs with 
diverging $\|H_{AB}^{(N)}\|$, belong 
to Class II.

For the Erd\H{o}s--R\'enyi graph, both conditions are satisfied 
rigorously. The first condition follows from 
the Wigner semicircle law for $H_{AA}$ and 
$H_{BB}$, the independence of $H_{AB}$ from 
$H_{BB}$ by the graph construction, and the 
bounded operator norm $\|H_{AB}\|= \mathcal{O}(1)$. The second 
condition follows from the quantum ergodicity 
of Wigner matrices~\cite{ErdosKnowlesYauYin2013},  so 
no eigenvector is supported entirely on $A$ or 
$B$. For graph 
ensembles where one or both conditions fail,
such as geometrically local graphs where the 
boundary-to-bulk ratio vanishes or scale-free graphs
where $\|H_{AB}\|\to\infty$ and the contour 
integral may break down, the volume law does 
not hold. A different scaling of 
entanglement entropy is expected in these cases.}

We have studied the ground-state entanglement entropy of free
fermions on Erd\H{o}s--R\'enyi random graphs in the dense regime, establishing that
the entropy obeys a volume law in the thermodynamic limit. We showed that the
entropy density is independent
of the edge probability and of the microscopic distribution
of the matrix entries.
The entropy density
interpreted as the entanglement per fermion, an intensive
thermodynamic quantity, 
that characterises how much quantum correlation each fermion
carries when embedded in a  scrambled non-local
environment~\cite{Page1993,Vidmar2017}. The strictly positive entropy deficit from the maximum value $\log 2$ achievable by a single fermionic
mode is itself a universal
constant, providing
a sharp characterisation of the departure of the
free-fermion ground state from maximal entanglement. The numerical value $s_\infty$ agrees to four decimal places 
with $2\ln 2 - 1$, 
leading to 
Conjecture~\ref{conj:arcsine}.

 A natural question is whether this entropy density depends only on the divergence of the average degree, rather than on the specific scaling  of  the edge probability. This raises the possibility that graph ensembles with $Np\rightarrow \infty$ but $p \rightarrow 0$, such as $p \sim 1/\log~N$, belong to the same entanglement universality class. Investigating this regime and identifying possible crossover behaviour remain important directions for future work.

Another central result of our analysis is that the bipartite entanglement entropy of free fermions on random graphs is not completely determined by the asymptotically free description of the block Hamiltonians. While the asymptotic-freeness approximation accurately captures the average spectral properties of the environment and reproduces the dominant contribution to the entropy, a \emph{finite but small} discrepancy between the exact and free-theory entropies persists in the thermodynamic limit. This demonstrates that entanglement entropy retains sensitivity to structural information beyond the asymptotically free spectral data. We find therefore entanglement entropy as a non-free-probabilistic observable, even in the thermodynamic limit, for random graphs.

Beyond its mathematical simplicity, the Erdős–Rényi ensemble serves as an idealized model of generic connectivity in disordered quantum systems. Similar random-network descriptions arise naturally in disordered quantum-dot arrays~\cite{cuadra_entropy_2021} and  probabilistic quantum communication networks~\cite{perseguers2010,kimble2008,Biroli2001}. Although the connectivity in real systems generally does not follow the Binomial distribution, our approach provides a technique to solve models of asymptotically asymmetric connectivity. Therefore the entropy density constitutes a theoretical prediction for the entanglement per fermion in any system within a universality class, against which experiments and numerical simulations of more complex models can be compared. In particular, the role of interactions, absent in the present
free-fermion baseline, in modifying the entropy remains to be studied in strongly correlated systems
in random graphs.

\bigskip

\bibliography{manuscript_p.bib}

\onecolumn

\clearpage
\setcounter{equation}{0}
\setcounter{lemma}{0}
\setcounter{theorem}{0}
\renewcommand{\theequation}{A\arabic{equation}}

\setcounter{section}{0}

\section*{APPENDIX A}\label{app:A}

\subsection*{Schur Complement Lemma}

\begin{lemma}
\label{lem:schur}
For all $z \in \mathbb{C}\setminus\mathbb{R}$, the $(A,A)$ block of the Green function
$G(z) = (zI_N -h_N)^{-1}$ is given by~\cite{zhang2005}
\begin{equation}
  \label{eq:schur}
  G_{AA}(z)
  =
  \left(zI_A -H_{AA}   - \Sigma_A(z)\right)^{-1},
\end{equation}
where the \emph{self-energy} matrix is
\begin{equation}
  \label{eq:selfenergy}
  \Sigma_A(z)
  =
  H_{AB}(zI_B - H_{BB})^{-1}\,H_{BA}.
\end{equation}
\end{lemma}

\begin{proof}
Since $H_{BB}$ is real symmetric and $\Im z < 0$, the matrix
$zI_B - H_{BB} $ is invertible. Applying block Gaussian elimination
to $(zI_N - h_N)$ we find
\begin{eqnarray}
 \begin{pmatrix}
   - H_{AA} + zI_A & -H_{AB} \\
    -H_{BA}        & -H_{BB} + zI_B
  \end{pmatrix}
  =
\begin{pmatrix}
    I_A & H_{AB}(-H_{BB} +  zI_B)^{-1} \\
    0   & I_B
  \end{pmatrix}
&&  \begin{pmatrix}
    S_A(z) & 0 \\
    -H_{BA} & -H_{BB} + zI_B
  \end{pmatrix}, \nonumber\\
\end{eqnarray}
where
\begin{equation}
  S_A(z)
  =
   zI_A - H_{AA} 
  -
  H_{AB}(zI_B - H_{BB})^{-1}H_{BA}
\end{equation}
 is the Schur complement of $(zI_B-H_{BB})$ in $(zI_N-h_N)$.
Since both factors are invertible, $S_A(z)$ is invertible.
Taking inverses and reading off the $(A,A)$ block yields
\[
G_{AA}(z)=S_A(z)^{-1},
\]
which is \eqref{eq:schur}.
\end{proof}
The Schur complement 
identity~\eqref{eq:schur} has an 
important consequence for the analytic 
structure of $G_{AA}(z)$. The poles 
of $G_{AA}(z)$ are the values of $z$ 
where $S_A(z)$ fails to be invertible, 
i.e. where $\det(S_A(z)) = 0$. Taking 
determinants of both sides of the 
block Gaussian elimination gives the 
identity
\begin{equation}
  \det(zI_N - H_N) 
  = \det(S_A(z))\cdot
    \det(zI_B - H_{BB}).
  \label{eq:block_det}
\end{equation}
Since $H_{BB}$ is real symmetric, all 
its eigenvalues are real, and therefore 
$\det(zI_B - H_{BB}) \neq 0$ for all 
$z\in\mathbb{C}\setminus\mathbb{R}$. 
Consequently
\begin{equation}
  \det(S_A(z)) = 0 
  \iff 
  \det(zI_N - H_N) = 0 
  \iff 
  z\in\sigma(H_N),
  \label{eq:poles}
\end{equation}
and since $H_N$ is real symmetric. 
Therefore all poles of $G_{AA}(z)$ 
are real, and any contour $\gamma$ 
with $|\mathrm{Im}(z)| = \eta > 0$ 
is a valid integration contour for 
the representation~\eqref{eq:CA_contour} 
of $C_A$.

\subsection*{Proof of Lemma 2}
\begin{proof}
(a) Write $z = x + i\eta$ with $\eta > 0$. 
Since $H_{AB}$ is real, the imaginary 
part passes through it
\begin{equation}
  \mathrm{Im}\,\Sigma^{(N)}(z) 
  = \mathrm{Im}\bigl(H_{AB}(zI_m - H_{BB})^{-1}
    H_{AB}^\top\bigr) 
  = H_{AB}\,
    \mathrm{Im}\bigl((zI_m - H_{BB})^{-1}\bigr)
    H_{AB}^\top.
  \label{eq:sigma_imag}
\end{equation}
Since $H_{BB}$ is Hermitian with real 
eigenvalues $\mu_j$ and orthonormal 
eigenvectors $u_j$, the resolvent 
decomposes as
\begin{equation}
  (zI_m - H_{BB})^{-1} 
  = \sum_j\frac{u_ju_j^\top}
    {(x-\mu_j)+i\eta},
  \label{eq:resolvent_decomp}
\end{equation}
whose imaginary part is
\begin{equation}
  \mathrm{Im}\bigl((zI_m - H_{BB})^{-1}\bigr) 
  = -\eta\sum_j
    \frac{u_ju_j^\top}{(x-\mu_j)^2+\eta^2} 
  \leq 0 
  \quad\text{for } \eta > 0.
  \label{eq:im_resolvent}
\end{equation}
Since congruence by $H_{AB}$ preserves 
negative semidefiniteness, 
$\eqref{eq:sigma_imag}$ 
and~$\eqref{eq:im_resolvent}$ give
\begin{equation}
  \mathrm{Im}\,\Sigma^{(N)}(z) 
  = H_{AB}\,
    \mathrm{Im}\bigl((zI_m - H_{BB})^{-1}\bigr)
    H_{AB}^\top 
  \leq 0.
  \label{eq:sigma_imag_neg}
\end{equation}

Define 
\begin{equation}
  A^{(N)}(z) = zI_m - H_{AA}^{(N)} 
- \Sigma^{(N)}(z),
\label{eq:A}
\end{equation}
so that 
$G_{AA}^{(N)}(z) = (A^{(N)}(z))^{-1}$. 
Since $H_{AA}^{(N)}$ is Hermitian, 
$\mathrm{Im}(H_{AA}^{(N)}) = 0$, and 
by~\eqref{eq:sigma_imag_neg}
\begin{equation}
  \mathrm{Im}\,A^{(N)}(z) 
  = \eta I_m 
  - \mathrm{Im}\,\Sigma^{(N)}(z) 
  \geq \eta I_m.
  \label{eq:imA}
\end{equation}
For any nonzero $v\in\mathbb{C}^m$
\begin{equation}
  \|A^{(N)}(z)v\|\cdot\|v\|
  \geq
  |\langle v, A^{(N)}(z)v\rangle|
  \geq
  \langle v,\mathrm{Im}\,A^{(N)}(z)\,v\rangle
  \geq \eta\|v\|^2,
  \label{eq:Av_lower}
\end{equation}
giving $\|A^{(N)}(z)v\| \geq \eta\|v\|$ 
for all $v$, and therefore:
\begin{equation}
  \|G_{AA}^{(N)}(z)\| 
  = \|(A^{(N)}(z))^{-1}\| 
  \leq \frac{1}{\eta},
  \label{eq:GAA_bound}
\end{equation}
uniformly for all $N$ and all $z$ with 
$|\mathrm{Im}(z)|\geq\eta$. The identical 
argument applied to 
$\bar{A}^{(N)}(z) = zI_m - H_{AA}^{(N)} 
- \bar\Sigma(z)$, using 
$\mathrm{Im}\,\bar\Sigma(z) = 
\lim_{N\to\infty}\mathrm{Im}\,\Sigma^{(N)}(z) 
\leq 0$, gives 
$\|\bar{G}_{AA}^{(N)}(z)\| \leq 1/\eta$.

(b) By \eqref{eq:A}, defining

\begin{equation}
   A^{(N)}(z) = zI_m - H_{AA}^{(N)} 
- \Sigma^{(N)}(z), \text{~~and~~}   \bar{A}^{(N)}(z) = zI_m - H_{AA}^{(N)} 
- \bar{\Sigma}(z),    
\end{equation}

we have

\begin{equation}
  \bar{A}^{(N)}(z) -  A^{(N)}(z) = \bar{\Sigma}(z)
- \Sigma^{(N)}(z).
\end{equation}

Using the identity
\begin{equation}
  (\bar{A}^{(N)}(z))^{-1} -  (A^{(N)}(z))^{-1} 
  = (\bar{A}^{(N)}(z))^{-1} \left(\bar{A}^{(N)}(z) - {A}^{(N)}(z)\right)  (A^{(N)}(z))^{-1},
\end{equation}

we get

\begin{equation}
  (\bar{A}^{(N)}(z))^{-1} -  (A^{(N)}(z))^{-1} 
  = (\bar{A}^{(N)}(z))^{-1} \left(\bar{\Sigma}(z)
- \Sigma^{(N)}(z)\right)  (A^{(N)}(z))^{-1},
\end{equation}

Since $H_{AA}^{(N)}$ appears in both 
$A^{(N)}(z)$ and $\bar{A}^{(N)}(z)$ and 
cancels in the difference 
$\bar{A}^{(N)}(z) - A^{(N)}(z) = 
\Sigma^{(N)}(z) - \bar\Sigma(z)$, the 
resolvent identity gives
\begin{equation}
  G_{AA}^{(N)}(z) - \bar{G}_{AA}^{(N)}(z)
  = G_{AA}^{(N)}(z)
    \bigl(\Sigma^{(N)}(z) - \bar\Sigma(z)\bigr)
    \bar{G}_{AA}^{(N)}(z).
  \label{eq:resolvent_diff}
\end{equation}
Taking operator norms and using 
$\eqref{eq:GAA_bound}$
\begin{equation}
 \lim_{N \to \infty} \sup_{|\mathrm{Im}(z)|\geq\eta}
  \|G_{AA}^{(N)}(z) - \bar{G}_{AA}^{(N)}(z)\|
  \leq \frac{1}{\eta^2}
  \sup_{|\mathrm{Im}(z)|\geq\eta}
  \|\Sigma^{(N)}(z) - \bar\Sigma(z)\|
= 0,
  \label{eq:final_rate}
\end{equation}
by assumption~\eqref{eq:sigma_conv}. 
\end{proof}

\subsection*{Proof of Theorem 1}
\begin{proof}
The proof proceeds through the following 
chain of implications.

\emph{(a) Green function convergence} 
By Lemma~\ref{lem:boundary_response}, 
the assumption $\Sigma^{(N)}(z) \to 
\bar\Sigma(z)$ implies 
$G_{AA}^{(N)}(z) \to \bar{G}_{AA}^{(N)}(z)$ 
uniformly for $|\mathrm{Im}(z)|\geq\eta$, 
with rate
\begin{equation}
  \sup_{|\mathrm{Im}(z)|\geq\eta}
  \|G_{AA}^{(N)}(z) - \bar{G}_{AA}^{(N)}(z)\|
  \leq \frac{1}{\eta^2}
  \sup_{|\mathrm{Im}(z)|\geq\eta}
  \|\Sigma^{(N)}(z) - \bar\Sigma(z)\|
  \to 0.
\end{equation}

\emph{(b) Correlation matrix 
convergence} By assumption (i), all 
poles of $G_{AA}^{(N)}(z)$ lie in the 
bounded real interval $[-M,M]$, so the 
fixed contour $\gamma$ encloses exactly 
the negative eigenvalues for every $N$. 
The contour integral then gives
\begin{equation}
  \|C_A^{(N)} - \bar{C}_A^{(N)}\| 
  \leq \frac{\ell(\gamma)}{2\pi}
  \sup_{z\in\gamma}
  \|G_{AA}^{(N)}(z) - \bar{G}_{AA}^{(N)}(z)\| 
  \to 0,
\end{equation}
where $\ell(\gamma)$ is the length of 
$\gamma$. Both $C_A^{(N)}$ and 
$\bar{C}_A^{(N)}$ are Hermitian, since 
$\gamma$ is symmetric about the real 
axis.

\emph{(c) Eigenvalue convergence.} 
By Weyl's inequality for Hermitian 
matrices
\begin{equation}
  \max_i|\lambda_i^{(N)} 
  - \bar\lambda_i^{(N)}| 
  \leq \|C_A^{(N)} - \bar{C}_A^{(N)}\| 
  \to 0.
\end{equation}

 Eigenvalue convergence implies the 
empirical spectral measures converge 
weakly, $\rho_{C_A^{(N)}} \Rightarrow 
\bar\rho_{C_A}$.

\emph{(d) Entropy convergence.} 
Since $h$ is continuous and bounded 
on $[0,1]$, weak convergence gives
\begin{equation}
  \frac{S_A}{m} 
  = \int_0^1 h(\lambda)\,
    d\rho_{C_A^{(N)}}(\lambda) 
  \xrightarrow{N\to\infty} 
  \int_0^1 h(\lambda)\,
    d\bar\rho_{C_A}(\lambda) 
  = s_\infty,
\end{equation}
which is positive by the assumption 
$s_\infty > 0$.

\emph{(e) Bounds on $s_\infty$} Since 
$C_A$ is a compression of the Fermi 
projector onto $m = N/2$ filled states, 
half-filling gives
\begin{equation}
  \int_0^1\lambda\,d\bar\rho_{C_A}(\lambda) 
  = \frac{1}{m}\operatorname{Tr}C_A 
  = \frac{1}{2}.
\end{equation}
The binary entropy $h(\lambda)$ is 
strictly concave on $[0,1]$ with unique 
maximum $h(1/2) = \log 2$, so by 
Jensen's inequality
\begin{equation}
  s_\infty 
  = \int_0^1 h(\lambda)\,
    d\bar\rho_{C_A}(\lambda) 
  \leq h\!\left(
    \int_0^1\lambda\,
    d\bar\rho_{C_A}(\lambda)
  \right) 
  = \log 2,
  \label{eq:Jensen}
\end{equation}
with equality if and only if 
$\bar\rho_{C_A} = \delta(\lambda-\tfrac{1}{2})$.

{\color{black} The entanglement entropy density 
$s_\infty$ measures how far the 
eigenvalues $\{\lambda_i\}$ of $C_A$ 
are spread away from $0$ and $1$. 
When all eigenvalues are exactly $0$ 
or $1$, $C_A$ is a projector, 
$h(\lambda_i) = 0$ for all $i$, and 
$s_\infty = 0$ . 
Therefore the condition $s_\infty > 0$ 
requires eigenvalues to be spread 
strictly into the interior of $(0,1)$
\begin{equation}
  \frac{1}{m}\operatorname{Tr}(C_A - C_A^2)
  = \frac{1}{m}\sum_{i=1}^m
    \lambda_i(1-\lambda_i) > 0,
\label{eq:idem_defect}
\end{equation}
The idempotency 
defect~\eqref{eq:idem_defect} measures 
the cross-boundary coherence of the 
occupied subspace, the degree to 
which the occupied eigenstates are 
delocalized across the bipartition. If the occupied subspace 
could be decomposed entirely into 
states supported on $A$ or on $B$, 
then $C_A$ would be a projector, the 
idempotency defect would vanish, and therefore
$s_\infty$ would vanish. A volume law therefore 
requires  the 
weaker condition of ergodicity that a finite fraction 
of the occupied subspace retains 
non-vanishing cross-boundary coherence 
in the thermodynamic limit. The 
idempotency defect is therefore the 
minimal criterion for the 
emergence of a positive entanglement 
entropy density.}
\end{proof}

As established from the Schur complement 
identity~\eqref{eq:schur} and the block 
determinant identity~\eqref{eq:block_det}, 
the poles of $G_{AA}^{(N)}(z)$ coincide 
with the real eigenvalues of $H_N$. In addition, the spectral radius of the Hamiltonianis 
bounded, $\|H_N\| \leq M$ from the assumption of Theorem~\ref{thm:criterion}.the  
Therefore any contour $\gamma$ with 
$|\mathrm{Im}(z)| = \eta > 0$ avoids 
all poles of $G_{AA}^{(N)}(z)$, and 
the resolvent bound 
$\|G_{AA}^{(N)}(z)\|\leq 1/\eta$ 
ensures the contour integral 
representation~\eqref{eq:CA_contour} 
is well defined.

\section*{APPENDIX B}\label{app:B}
\subsection*{Spectral Distributions of $H_{AA}$ and $H_{BB}$}

The blocks $H_{AA}$ and $H_{BB}$ are each $(N/2)\times(N/2)$
Wigner-type random matrices with centred independent entries
of variance $1/N$.
Since their size is $M=N/2$ but their entry variance is
$1/N = 2/M\cdot(1/2)$, the effective variance parameter in
units of the block size is $\sigma_*^2 = 1/2$.
By the Wigner semicircle law~\cite{ErdosKnowlesYauYin2013},
the empirical spectral distributions of $H_{AA}$ and $H_{BB}$
each converge to the semicircle law
\begin{equation}
  \mu(x) = \frac{1}{\pi}\sqrt{2-x^2},
  \qquad x\in[-\sqrt2,\sqrt2].
  \label{eq:semicircle}
\end{equation}
The corresponding Stieltjes transform is
\begin{equation}
  m_{AA}(z) = m_{BB}(z)
  = \int_{-\sqrt2}^{\sqrt2} \frac{\mu(x)}{z-x} dx,
  \qquad z\in\mathbb{C}\setminus[-\sqrt2,\sqrt2].
  \label{eq:m_def}
\end{equation}
 
The empirical spectral distributions of $H_{AA}$ and $H_{BB}$ converge
to the semicircle law
\begin{equation}
\mu(x)=\frac{1}{\pi}\sqrt{2-x^2},
\qquad x\in[-\sqrt2,\sqrt2].
\end{equation}

Introducing the rescaled variables
$x=\sqrt2\,u$ and $w=z/\sqrt2$, one obtains
\begin{equation}
m(z)
=
\sqrt2\,S(w),
\qquad
S(w)
=
\frac{1}{\pi}
\int_{-1}^{1}
\frac{\sqrt{1-u^2}}{w-u}\,du .
\label{eq:S_def}
\end{equation}
The integral $S(w)$ is the Stieltjes transform of the standard
semicircle distribution on $[-1,1]$. Evaluating it, for example by
setting $u=\cos\theta$ and reducing the resulting expression to the
classical integral
\[
\int_0^\pi \frac{d\theta}{w-\cos\theta}
=
\frac{\pi}{\sqrt{w^2-1}},
\qquad
w\notin[-1,1],
\]
gives
\begin{equation}
S(w)=w-\sqrt{w^2-1},
\end{equation}
where the branch is chosen so that
$S(w)\sim (2w)^{-1}$ as $|w|\to\infty$. Undoing the rescaling yields
\begin{equation}
m_{AA}(z)=m_{BB}(z)
=
z-\sqrt{z^2-2},
\label{eq:mAA_mBB}
\end{equation}
which is the expression used in the main text.

\subsection*{Spectral Distribution of $W = H_{AB}H_{AB}^{\top}$}

We define the $(N/2)\times(N/2)$ matrix
\begin{equation}
W = H_{AB}H_{AB}^{\top}.
\end{equation}
Since the entries of $H_{AB}$ have variance $1/N$, we may write
\begin{equation}
H_{AB}=N^{-1/2}\widetilde X,
\end{equation}
where $\widetilde X$ has independent centred entries of unit variance.
It follows that
\begin{equation}
W
=
\frac{1}{N}\widetilde X\widetilde X^{\top}
=
\frac12 Y,
\qquad
Y
=
\frac{\widetilde X\widetilde X^{\top}}{N/2}.
\end{equation}
The matrix $Y$ is a sample covariance matrix with aspect ratio
$c=1$. By the Marchenko--Pastur theorem~\cite{marchenko_1967},
its empirical spectral distribution converges to
\begin{equation}
\rho_Y(\lambda)
=
\frac{1}{2\pi}
\sqrt{\frac{4-\lambda}{\lambda}},
\qquad
\lambda\in(0,4].
\end{equation}

Since $W=\frac12Y$, the spectral density rescales according to
\begin{equation}
\mu_W(x)
=
2\,\rho_Y(2x)
=
\frac{1}{\pi}
\sqrt{\frac{2-x}{x}},
\qquad
x\in(0,2].
\label{eq:MP_density}
\end{equation}

The corresponding Stieltjes transform is
\begin{equation}
m_W(z)
=
\int_0^2
\frac{\mu_W(x)}{z-x}\,dx
=
\frac1\pi
\int_0^2
\frac{1}{z-x}
\sqrt{\frac{2-x}{x}}\,dx .
\label{eq:mW_def}
\end{equation}
Introducing the substitution $x=1-\cos\phi$ reduces the integral to
\begin{equation}
m_W(z)
=
\frac1\pi
\int_0^\pi
\frac{1+\cos\phi}{z-1+\cos\phi}\,d\phi,
\end{equation}
which may be evaluated using the standard Poisson-kernel integral. The result is
\begin{equation}
m_W(z)
=
1-\sqrt{\frac{z-2}{z}},
\label{eq:mW}
\end{equation}
where the branch of the square root is chosen such that
$m_W(z)\sim z^{-1}$ as $|z|\to\infty$, as required for a
Stieltjes transform.

\section*{APPENDIX C}
\label{app:p_independence}
Consider the normalized Erd\H{o}s--R\'enyi Hamiltonian
\begin{equation}
H_{N,p}
=
\frac{A_N-p(\mathbf 1\mathbf 1^\top-I_N)}
{\sqrt{Np(1-p)}}.
\label{eq:H_scaling}
\end{equation}

We show that the thermodynamic entanglement entropy density is independent of \(p\in(0,1)\).

The $i$-th component of $t_N\mathbf{1}$ is the 
degree of vertex $i$
\begin{equation}
  (t_N\mathbf{1})_i 
  = \sum_{j\neq i}A_{ij} 
  = \deg(i).
  \label{eq:degree}
\end{equation}
In the Erd\H{o}s--R\'enyi graph $G(N,p)$, the 
degree of each vertex satisfies 
$\deg(i)\sim\mathrm{Binomial}(N-1,p)$, with
\begin{equation}
  \mathbb{E}[\deg(i)] = (N-1)p \approx Np, 
  \qquad 
  \mathrm{Var}[\deg(i)] = (N-1)p(1-p) = O(N).
\end{equation}
By the law of large numbers, $\deg(i)/N \to p$ 
almost surely as $N\to\infty$, so
\begin{equation}
  (t_N\mathbf{1})_i 
  = Np + O(\sqrt{N}) 
  \quad \text{almost surely, for each } i.
  \label{eq:degree_concentration}
\end{equation}
Therefore
\begin{equation}
  t_N\mathbf{1} 
  = Np\cdot\mathbf{1} 
  + \mathcal{O}(\sqrt{N})\cdot\mathbf{1} 
  \quad \text{almost surely,}
  \label{eq:tN1}
\end{equation}
and dividing by $\sqrt{N}$
\begin{equation}
  t_N\frac{\mathbf{1}}{\sqrt{N}} 
  = Np\cdot\frac{\mathbf{1}}{\sqrt{N}} 
  + \mathcal{O}(\sqrt{N})\cdot\frac{\mathbf{1}}{\sqrt{N}}.
  \label{eq:tN_outlier}
\end{equation}
The relative magnitude of the error term is
\begin{equation}
  \frac{\left\|\mathcal{O}(\sqrt{N})\cdot
  \frac{\mathbf{1}}{\sqrt{N}}\right\|}{Np} 
  = \frac{\mathcal{O}(\sqrt{N})\cdot
  \frac{\sqrt{N}}{\sqrt{N}}}{Np} 
  = \frac{\mathcal{O}(\sqrt{N})}{Np} 
  = \mathcal{O}\left(\frac{1}{\sqrt{N}}\right) 
  \to 0
  \label{eq:relative_error}
\end{equation}
as $N\to\infty$. Therefore $\mathbf{1}/\sqrt{N}$ 
is an asymptotically exact eigenvector of $t_N$ 
with eigenvalue
\begin{equation}
  \lambda_{\mathrm{out}}^{(t)} 
  \approx Np \to \infty 
  \quad \text{as } N\to\infty.
  \label{eq:outlier_eigenvalue}
\end{equation}

The bulk spectrum of $t_N$ lies in the interval 
$[-2\sqrt{Np(1-p)}, 2\sqrt{Np(1-p)}]$ by the 
Wigner semicircle law applied to the fluctuation 
matrix $t_N - \mathbb{E}[t_N]$~\cite{ErdosKnowlesYauYin2013}. 
The spectral gap between the outlier eigenvalue 
$\lambda_{\mathrm{out}}^{(t)} \approx Np$ and the 
top of the bulk spectrum $2\sqrt{Np(1-p)}$ is
\begin{equation}
  \lambda_{\mathrm{out}}^{(t)} 
  - 2\sqrt{Np(1-p)} 
  \approx Np - 2\sqrt{Np(1-p)} 
  = \mathcal{O}(N) \to \infty.
  \label{eq:spectral_gap}
\end{equation}
The perturbation 
of each bulk eigenvector $\psi_k$ due to the 
outlier direction $\mathbf{1}/\sqrt{N}$ is bounded 
by the ratio of the perturbation norm to the 
spectral gap
\begin{equation}
  \left|\left\langle\frac{\mathbf{1}}{\sqrt{N}}
  \bigg|\psi_k\right\rangle\right| 
  \leq \frac{\mathcal{O}(\sqrt{N})}{\mathcal{O}(N)} 
  = \mathcal{O}\left(\frac{1}{\sqrt{N}}\right) 
  \to 0
  \label{eq:DavisKahan}
\end{equation}
as $N\to\infty$. Therefore in the thermodynamic 
limit
\begin{equation}
  \left\langle\frac{\mathbf{1}}{\sqrt{N}}
  \bigg|\psi_k\right\rangle = 0, 
  \qquad k = 1,\ldots,N-1.
  \label{eq:orthogonality}
\end{equation}
For any bulk eigenvector $\psi_k$ 
satisfying~\eqref{eq:orthogonality}
\begin{equation}
  p\mathbf{1}\mathbf{1}^\top\psi_k 
  = p\mathbf{1}(\mathbf{1}^\top\psi_k) 
  = p\mathbf{1}\cdot\sqrt{N}
  \left\langle\frac{\mathbf{1}}{\sqrt{N}}
  \bigg|\psi_k\right\rangle 
  = p\mathbf{1}\cdot\sqrt{N}\cdot 0 
  = 0.
  \label{eq:rank1_vanishes}
\end{equation}

The normalized Hamiltonian is
\begin{equation}
  H_{N,p} 
  = \frac{t_N - p(\mathbf{1}\mathbf{1}^\top - I_N)}
  {\sqrt{Np(1-p)}} 
  = \frac{t_N - p\mathbf{1}\mathbf{1}^\top + pI_N}
  {\sqrt{Np(1-p)}}.
  \label{eq:HNp_decomposed}
\end{equation}
For any bulk eigenvector $\psi_k$ of $t_N$ with 
eigenvalue $\lambda_k^{(t)}$, using 
\eqref{eq:rank1_vanishes}
\begin{equation}
  (t_N - p\mathbf{1}\mathbf{1}^\top + pI_N)\psi_k 
  = t_N\psi_k 
  - p\mathbf{1}\mathbf{1}^\top\psi_k 
  + p\psi_k \nonumber\\
  = \lambda_k^{(t)}\psi_k - 0 + p\psi_k 
  \nonumber\\
  = (\lambda_k^{(t)} + p)\psi_k.
  \label{eq:shifted_eig}
\end{equation}
Applying the scalar scaling 
$1/\sqrt{Np(1-p)} > 0$
\begin{equation}
  H_{N,p}\psi_k 
  = \frac{\lambda_k^{(t)} + p}{\sqrt{Np(1-p)}}
  \psi_k 
  \equiv \varepsilon_k^{(H)}\psi_k,
  \label{eq:H_eig}
\end{equation}
where
\begin{equation}
  \varepsilon_k^{(H)} 
  = \frac{\lambda_k^{(t)} + p}{\sqrt{Np(1-p)}}
  = \frac{\lambda_k^{(t)}}{\sqrt{Np(1-p)}} 
  + \frac{p}{\sqrt{Np(1-p)}}.
  \label{eq:eigenvalue_relation}
\end{equation}
The second term vanishes in the thermodynamic 
limit
\begin{equation}
  \frac{p}{\sqrt{Np(1-p)}} 
  = \sqrt{\frac{p}{N(1-p)}} 
  = \mathcal{O}\left(\frac{1}{\sqrt{N}}\right) 
  \to 0 
  \quad \text{as } N\to\infty.
  \label{eq:shift_vanishes}
\end{equation}
Therefore in the thermodynamic limit
\begin{equation}
  \varepsilon_k^{(H)} 
  \xrightarrow{N\to\infty} 
  \frac{\lambda_k^{(t)}}{\sqrt{Np(1-p)}},
  \label{eq:eigenvalue_limit}
\end{equation}
and the bulk eigenvectors of $t_N$ and $H_{N,p}$ 
are identical in the thermodynamic 
limit. Since the bulk eigenvectors of $t_N$ and $H_{N,p}$ 
are identical and the occupied sets are 
identical, the Fermi projectors are equal 
in the thermodynamic limit
\begin{equation}
  P_F^{(t)} 
  = \sum_{\lambda_k^{(t)} < \mu_{\mathrm{orig}}}
  |\psi_k\rangle\langle\psi_k| 
  = \sum_{\varepsilon_k^{(H)} < 0}
  |\psi_k\rangle\langle\psi_k| 
  = P_F^{(H)}.
  \label{eq:projectors_equal}
\end{equation}
The entanglement correlation matrices are 
therefore identical
\begin{equation}
  C_A^{(t)} 
  = P_F^{(t)}\big|_A 
  = P_F^{(H)}\big|_A 
  = C_A^{(H)}.
  \label{eq:CA_equal}
\end{equation}
This confirms that the $p$-independence of 
$s_\infty$ is a fundamental property of the 
Erd\H{o}s--R\'enyi random graph ensemble, 
independent of the normalization convention 
used for the Hamiltonian. 

\end{document}